\documentclass[11pt,a4paper,twoside]{article}

\usepackage[T1]{fontenc}
\usepackage[cp1250]{inputenc} 
\usepackage{times}

\usepackage[leqno]{amsmath} 
\usepackage{amsthm,amsfonts} 
\usepackage{bm} 

\input xy  \xyoption{all}

\usepackage{geometry}
\usepackage[usenames]{color} 

\usepackage{graphicx} 

\usepackage{enumerate} 
\usepackage{hyperref} 
\newcommand{\new}{}

\newcommand{\TT}[1]{\mathcal{T}^{#1}} 
\newcommand{\dd}[1]{\mathrm{d}^{#1}} 
\newcommand{\dv}{\dd{}_{\V}} 
\newcommand{\Cf}{C^\infty} 
\newcommand{\X}{\vec{\mathcal{X}}} 
\newcommand{\RR}[1]{\mathcal{R}_{#1}} 
\newcommand{\Ann}{\operatorname{Ann}} 
\newcommand{\eps}[1]{\mathcal{E}_{#1}} 
\newcommand{\ddt}{\frac{\dd{}}{\dd {}t}}

\newcommand{\PP}{P} 
\newcommand{\bundle}{\sigma} 
\newcommand{\anchor}{\rho} 
\newcommand{\dual}[1]{{#1}^\ast} 
\newcommand{\tangent}[1]{ \tau_{#1}} 
\newcommand{\cotangent}[1]{\tau_{#1}^\ast} 
\newcommand{\vertical}[1]{\upsilon_{#1}} 
\newcommand{\fibration}{\pi} 
\newcommand{\lift}[2]{\TT{#2}{#1}} 
\newcommand{\liftanch}[2]{#1^{#2}} 
\newcommand{\down}[2]{#1_{#2}} 
\newcommand{\morphlift}[2]{\mathcal{T}^{#1}{#2}} 
\newcommand{\I}[1]{i_{#1}} 
\newcommand{\EL}{E_L} 

\newcommand{\V}{\mathrm{V}}

\DeclareMathAlphabet{\mathpzc}{OT1}{pzc}{m}{it}

\newcommand{\R}{\mathbb{R}}
\newcommand{\Sec}{\operatorname{Sec}}
\newcommand{\id}{\operatorname{id}}

\newcommand{\wt}[1]{\widetilde{#1}}
\newcommand{\T}{\mathrm{T}}
\newcommand{\lra}{\longrightarrow}
\newcommand{\ra}{\rightarrow}
\newcommand{\gra}{\operatorname{graph}}
\def\<#1>{\big\langle #1\big\rangle}
\def\(#1){\left( #1\right)}

\newdir{|>}{!/4,5pt/@{|}*:(1,-.2)@^{>}*:(1,+.2)@_{>}}

\numberwithin{equation}{section} 

\theoremstyle{plain} 
\newtheorem{theorem}{Theorem}[section]
\newtheorem{proposition}[theorem]{Proposition}
\newtheorem{lemma}[theorem]{Lemma}
\newtheorem{lemma_def}[theorem]{Lemma/Definition}

\theoremstyle{definition}
\newtheorem{definition}[theorem]{Definition}

\theoremstyle{remark}
\newtheorem{remark}[theorem]{Remark}
\newtheorem{corollary}[theorem]{Corollary}

\title{Prolongations vs. Tulczyjew triples in Geometric Mechanics\footnote{This research was supported by the Polish National Science Center (grant DEC-2012/06/A/ST1/00256).}}
\author{Michał Jóźwikowski\footnote{Email: \texttt{mjoz(at)impan.pl}} \\[0.5cm]
\emph{Institute of Mathematics}\\
\emph{Polish Academy of Sciences}
}
\date{\today}

\begin{document}
\maketitle
\begin{abstract}
In the scientific literature there are basically two schools of formulating Lagrangian (or Hamiltonian) mechanics in the (Lie) algebroid setting: in terms of prolongations and in terms of Tulczyjew triples. Despite the fact that in both approaches we describe the same phenomena, so far no comparison between prolongations and Tulczyjew triples was made. In this note we aim to fill this gap. More precisely, we will strip the prolongation approach to uncover the Tulczyjew triple reality hidden inside, thus proving that the latter approach is a more basic one.  
\end{abstract}

\section{Introduction}

\paragraph{A bit of history}
The career of algebroids in mechanics was originated by Weinstein \cite{Weinstein_1996} (see also \cite{Liberman_1996}) who posted a problem of formulating Analytic Mechanics in the language of Lie algebroids. The first solution was due to Mart\'{i}nez \cite{Martinez_1999,Martinez_2001} who adopted much earlier ideas of J. Klein \cite{Klein_1962} by using \emph{prolongations} of Lie algebroids introduced in \cite{Higgins_Mackenzie_1990} (under the name induced or inverse-image algebroids). This point of view soon spread among researchers working in Geometric Mechanics which resulted in translating almost every aspect of the classical theory (vakonomic and nonholonomic constraints, Hamilton-Jacobi theory, controls, higher order systems, etc.) into this new geometric setting (see \cite{Martinez_2007, Cortes_de_Leon_Marrero_Martinez_2009, Iglesias_Marrero_Diego_Sosa_2008, Colombo_2017} to name just a few references). Another approach to original Weinstein's problem, which can be traced back to the ideas of Tulczyjew \cite{Tulczyjew_1976a,Tulczyjew_1976b}, appeared in \cite{GGU_2006}. Is makes use of the structure called a \emph{Tulczyjew triple}, and also has been widely adopted in various aspects of mechanics and field theory -- see \cite{GG_2013}.

\paragraph{Two approaches to Geometric Mechanics}
As mentioned above, there are basically two schools of formulating the Hamiltonian and the Lagrangian dynamics on algebroids. In the \emph{Tulczyjew triple approach} one encodes the algebroid structure on a vector bundle $\bundle:E\ra M$ as canonical double vector bundle morphisms $\wt{\Lambda_{E^\ast}}:\T^\ast E^\ast\ra\T E^\ast$ and $\eps E:\T^\ast E\ra \T E^\ast$. These, together with the differential of a proper generating function $\dd{} H:E^\ast\ra \T^\ast E^\ast$ (in the Hamiltonian case) or $\dd{} L:E\ra\T^\ast E$ (in the Lagrangian case), are used to generate the phase dynamics, i.e a differential equation for a curve of momenta -- see \eqref{diag:TT}.  

In contrast, the \emph{prolongation approach} makes use of the $E^\ast$- and $E$-prolongations of the initial algebroid, i.e. bundles $\lift{\bundle}{E^\ast}:\TT{E^\ast}E\ra E^\ast$ and $\lift{\bundle}{E}:\TT{E}E\ra E$, in the Hamiltonian and Lagrangian cases, respectively. In the first case there exists a canonical non-degenerate 2-form $\Omega_E\in\Sec(\Lambda^2(\TT{E^\ast}E)^\ast)$, which is used to formulate the equations of motion in the symplectic terms. In the Lagrangian case one pull-backs $\Omega_E$ via $\morphlift{\lambda_L}E:\TT EE\ra \TT{E^\ast}E$ -- the prolongation of the Legendre map $\lambda_L:E\ra E^\ast$ (i.e. the vertical derivative of the Lagrangian $L:E\ra\R$). The related dynamics is formulated as a presymplectic dynamics related with this pull-back $\omega_L:=(\morphlift{\lambda_L}E)^\ast\Omega_E\in\Sec(\Lambda^2(\TT E E)^\ast)$. Note that $\omega_L$ is non-canonical (it depends on $L$) and possibly degenerated. 

\paragraph{Motivations and purposes}
It seems that the two presented approaches are of different geometric nature. It may though come as a surprise that they actually lead to the same equations of motion. It is, however, much more surprising that so far no one explained why it is so! In fact, researchers working in the field usually favor one approach (depending on their background) and refer to the competitive one by checking if local expressions match in both cases. In this work \textbf{we aim to understand what is a geometric relation between the Tulczyjew triple and the prolongation approaches to mechanics on algebroids}. In other words, we want to \textbf{understand why the obtained equations are the same}, rather than just conclude that they do. As Geometric Mechanics is already well-established it seems that our work will have little impact on future research, yet we feel that this is a problem should be addressed for the consistency of the theory, but was somehow overlooked by the community. 

\paragraph{Our results}
Our solution to the problem formulated above is quite simple and requires only a basic understanding of the notion of a prolongation. Given a fibration $\fibration:\PP\ra M$, one constructs $\TT \PP E$ -- the $E$-tangent bundle of $\PP$,  which is a subject to a canonical embedding  $\I \PP:\TT \PP E\hookrightarrow E\times_M\T \PP$ as a vector subbundle over $\PP$. The dual of this embedding gives rise to a vector bundle morphisms $\dual{(\I \PP)}:E^\ast\times \T^\ast \PP\ra (\TT {\PP} E)^\ast$. Now in the construction of Hamiltonian equations in the prolongation approach we actually use the vector bundle isomorphism $(\wt{\Omega_E})^{-1}:(\TT{E^\ast} E)^\ast\ra \TT{E^\ast} E$ induced by the inverse of $\Omega_E$. This, in turn, gives rise to a map $\dual{(\I{E^\ast})}\circ(\wt{\Omega_E})^{-1}\circ\I{E^\ast}:E^\ast\times \T^\ast E^\ast\ra E\times \T E^\ast$ which has a simple interpretation in terms of the Hamiltonian part of the Tulczyjew triple $\wt{\Lambda_{E^\ast}}:\T^\ast E^\ast \ra \T E^\ast$ (Theorem~\ref{thm:hamilton}). Actually, spaces $E^\ast$ and $E$ in the domain and co-domain of the composition $\dual{(\I{E^\ast})}\circ(\wt{\Omega_E})^{-1}\circ\I{E^\ast}$ are not important from the point of view of mechanics and we are left with the bare DVB morphisms $\wt{\Lambda_{E^\ast}}$ (i.e., the Hamiltonian side of the Tulczyjew triple). 

The Lagrangian formalism in the prolongation approach is formulated in terms of a 2-form $\omega_L$ introduced above. With $\omega_L$ me may associate the related vector-bundle morphism $\wt{\omega_L}:\TT EE\ra (\TT E E)^\ast$ and the dynamics is generated by the differential of the energy function $\dd{}\EL:E\ra \T^\ast E$. The presence of the later can be explained as a naturally induced by the Legendre map $\lambda_L$ from the differential of the Lagrangian $\dd{} L:E\ra \T^\ast E$ (Lemma~\ref{lem:L_to_EL}). This observation, together with the above results about the Hamiltonian formalism, allows to recognize the Lagrangian side of the Tulczyjew triple $\eps E:\T^\ast E\ra\T E^\ast$ hidden inside the formalism (Theorem \ref{thm:lagrangian}).   

A rough summary of the above results would be the following. The Tulczyjew triple approach is hidden inside the prolongation approach. It can be recovered by restricting and projecting maps $\wt{\Omega_E}$ and $\wt{\omega_L}$ to the proper ``legs'' of the prolongations $\TT{E^\ast}E$ and $\TT{E} E$. This proves that the Tulczyjew triple approach is a more fundamental one -- see discussion in Section~\ref{sec:concl}.   

\paragraph{Final remarks}
After introduction of basic concepts and definitions in Subsection~\ref{ssec:preliminaries}, we give a brief description of the Tulczyjew triple approach in Section~\ref{sec:TT} based mostly on \cite{GGU_2006}.
The treatment of prolongations and the related approach to geometric mechanics in Section~\ref{sec:comparison} is based in principle on \cite{CLM_2006}, yet at some points of the presentation, where we think that the theory can be simplified, we decided to detour from this survey. We claim originality for Lemma~\ref{lem:prolong} (it seems to us that such a characterization of the algebroid structure on the prolongation was not present in the literature), Theorem~\ref{thm:hamilton} (although the coordinate formula of $\wt{\Omega_E}$ appeared in \cite{LMM_2005}), Theorem~\ref{thm:lagrangian}, and an alternative formulation of mechanics in terms of the prolongation Tulczyjew triple in  Subsection~\ref{ssec:prolong_tt} (the triple itself was, however, known to \cite{LMM_2005}). 
 
To keep this work concise we omitted proofs if they can be found in the literature, discussing examples and giving coordinate descriptions. The quoted references will provide plenty of each for an interested reader.

\subsection{Preliminaries: double vector bundles and algebroids}
\label{ssec:preliminaries}

\paragraph{Basic geometric objects}
In the course of our considerations we will be working with a vector bundle (\emph{VB}, in short) $\bundle:E\ra M$ and its dual $\dual{\bundle}:E^\ast \ra M$. An application of the tangent $\T$ or the cotangent $\T^\ast$ functor to either of these makes them examples of \emph{double vector bundles} (\emph{DVB}s, in short). We will not need any specific properties of such objects (an interested reader should consult \cite{Konieczna_Urbanski_1999,Grabowski_Rotkiewicz_2009}) -- for our purposes it is enough to know that a DVB is a pair of naturally compatible VB structures sharing the same total space. For instance, on $\T E$ the related structures are: the tangent bundle structure $\tangent E:\T E\ra E$ and the tangent lift of the initial VB structure $\T\bundle:\T E\ra \T M$. In the case of $\T^\ast E$ (and similarly $\T^\ast E^\ast$) one of these two VB structures is the cotangent fibration $\cotangent E:\T^\ast E\ra E$. The second one, denoted $\vertical E:\T^\ast E\ra E^\ast$, seems not so obvious at first. We may interpret it as the \emph{vertical derivative}, namely if $f:E\ra \R$ is a smooth function representing a covector $\dd{}f(e)\in \T^\ast_e E$, then its image under $\vertical E$ evaluated on an element $e'\in E_{\tau(e)}$ is the derivative of $f$ in the direction of a vertical vector $\V_ee'\in\T_eE$. That is,   
$$\<\vertical E(\dd{}f(e)),e'>=\V_ee'(f)\overset{df}=\frac{\dd{}}{\dd{}t}\Big|_{t=0} f(e+te')\ .$$
For convenience we shall denote the vertical (fiber-wise) derivative by $\dv f:=\vertical E(\dd{}f)$.

\paragraph{The canonical DVB isomorphism}
Interestingly, spaces $\T^\ast E$ and $\T^\ast E^\ast$ are canonically isomorphic as DVBs. {\new The construction of this isomorphism, denoted $\RR E:\T^\ast E\ra \T^\ast E^\ast$, is based on the following construction of a Lagrangian submanifold on the cotangent bundle. Let $Q$ be a smooth manifold, $C\subset Q$ a smooth submanifold, and $f:C\ra \R$ a smooth function. Define the set
$$C_f:=\{\theta\in \T^\ast Q\ |\ \cotangent{Q}(\theta)\in C,\ \theta\big|_{\T C}=\dd{}f \}\ , $$
consisting of all covectors in $\T^\ast Q$ which behave as $\dd{} f$ when restricted to $\T C$. The restriction of the cotangent projection $\pi:=\cotangent{Q}\big|_{C_f}:C_f\ra C$ makes $C_f$ an affine bundle over $C$ modeled on the annihilator subbundle $\Ann(\T C)\subset \T^\ast Q\big|_{C}$. It is well-known that $C_f$ is a Lagrangian submanifold in $\T^\ast Q$ (note that the tautological 1-form $\theta_Q\in\Omega^1(\T^\ast Q)$ restricted to $C_f$ is simply $\theta_Q\big|_{C_f}=\pi^\ast\dd{} f$).    

Now consider the above construction of $C_f$ for $Q:=E\times E^\ast$, $C:=E\times_M E^\ast$ and $f:=\<\cdot,\cdot>:E\times_M E^\ast\ra \R$ 
being the canonical pairing. We define the graph of $\RR E$ to be the related Lagrangian submanifold, i.e. $\gra(\RR E):=C_f\subset\T^\ast (E\times E^\ast)\simeq \T^\ast E\times\T^\ast E^\ast$.}
We conclude that $\RR{E}$ is  also an anti-symplectomorphism. In the course of this paper we will use the fact that $\RR E$ intertwines the ``legs'' of DVBs $\T^\ast E$ and $\T^\ast E^\ast$, i.e.
\begin{equation}
\label{eqn:R_E_legs}
\vertical{E^\ast}\circ\RR E=\cotangent{E}:\T^\ast E\ra E\quad\text{and}\quad \cotangent{E^\ast}\circ\RR E=\vertical E:\T^\ast E\ra E^\ast\ .
\end{equation}

\paragraph{Almost-Lie algebroids} 
The fundamental definition of this note is the following one.
\begin{definition}
\label{def:algebroid}
An \emph{almost-Lie algebroid} (\emph{AL algebroid}, in short) structure on a VB $\bundle:E\ra M$ is constituted by a VB map $\anchor:E\ra \T M$ over $\id_M$ (the \emph{anchor}) and a skew-symmetric $\R$-bilinear operator $[\cdot,\cdot]:\Sec(E)\times\Sec(E)\ra\Sec(E)$ (the \emph{bracket}) satisfying the following \emph{Leibniz rule}
\begin{equation}
\label{eqn:leibniz}
[e,f\cdot e']=f[e,e']+\anchor(e)f\cdot e'\ ,
\end{equation}
and the condition of \emph{compatibility between the anchor and the bracket}
\begin{equation}
\label{eqn:AL}
\anchor[e,e']=[\anchor(e),\anchor(e')]_{\T M}
\end{equation}
for any $e,e'\in \Sec(E)$ and $f\in\Cf(M)$. Here $[\cdot,\cdot]_{\T M}$ denotes the standard Lie bracket of vector fields on $M$. 
\end{definition}
Typically in geometric mechanics one works with \emph{Lie algebroids}, which are AL algebroids with the bracket additionally satisfying the \emph{Jacobi identity}. In our considerations the Jacobi identity will play no role, so we chose to work in greater generality. However, we cannot get rid of axiom \eqref{eqn:AL}, which is necessary to equip the prolongation of an algebroid with the canonical algebroid structure (see Lemma~\ref{lem:prolong}). On the other hand, more general classes of algebroid structures were introduced in \cite{GU_1999}.

\paragraph{Algebroids as derivations}
In the prolongation approach to geometric mechanics another characterization of an algebroid structure is frequently used.

\begin{proposition}
\label{prop:alg_dif} 
An AL algebroid structure $(\bundle:E\ra M,\anchor,[\cdot, \cdot])$ can be equivalently characterized as a degree-one derivation $\dd E$ in the graded algebra of $E$-forms $(\Sec(\Lambda^\bullet(E^\ast)),\wedge)$. It is determined by its action on zero and one $E$-forms by the following generalization of the Cartan formula 
\begin{equation}
\label{eqn:cartan_formula}
\dd Ef(e):=\anchor(e)f\quad\text{and}\quad \dd E\xi(e,e'):=\anchor(e)\<\xi,e'>-\anchor(e')\<\xi,e>-\<\xi,[e,e']>\ .
\end{equation}
Here $f\in\Cf(M)$, $e,e'\in\Sec(E)$, and $\xi\in \Sec(E^\ast)$. 
\end{proposition}
It is well-known that a Lie algebroid structure corresponds to the case of a homological derivation, i.e. $\dd E\dd E=0$. On the other hand, axiom \eqref{eqn:AL} is equivalent to a condition that $\dd E\dd E f\equiv 0$ for every function $f\in\Cf(M)$. 

The above characterization allows for an elegant definition of a morphism in the category of algebroids.
\begin{definition}
\label{def:algebroid_morphism}
Let $(\bundle:E\ra M,\anchor,[\cdot,\cdot])$ and $(\bundle':E'\ra M',\anchor',[\cdot,\cdot]')$ be two AL algebroids. We say that a VB map $\psi:E\ra E'$ is an \emph{algebroid morphism} if the induced map $\psi^\ast:\Sec(\Lambda^\bullet(E')^\ast)\ra\Sec(\Lambda^\bullet(E^\ast))$ commutes with the derivations $\dd E$ and $\dd{E'}$, i.e.
$$\psi^\ast(\dd{E'}\Xi)=\dd E(\psi^\ast \Xi)\ ,$$
for every $\Xi\in\Sec(\Lambda^k(E')^\ast)$, $k=0,1,2,\hdots$. Actually, it suffices to check this condition only for $k=0,1$. 
\end{definition}
In particular, axiom \eqref{eqn:AL} can be viewed as the condition that the anchor $\anchor$ is an algebroid morphism between $\bundle:E\ra M$ endowed with the original algebroid structure and $\tangent M:\T M \ra M$ with the canonical tangent bundle algebroid structure. 

\paragraph{Algebroids as linear bi-vector fields} 
On the other hand, in the Tulczyjew triple approach to mechanics the following point of view is useful.
\begin{proposition}[\cite{GU_1999}]
\label{prop:alg_bi_vect}
An AL algebroid structure $(\bundle:E\ra M,\anchor,[\cdot, \cdot])$ can be equivalently characterized as a bi-vector field $\Lambda_{E^\ast}\in\frak{X}^2(E^\ast)$. It is determined by the following formulas 
\begin{equation}
\label{eqn:lambda_formula}
\Lambda_{E^\ast}(\dd{}\iota_e,(\dual{\bundle})^\ast\dd{} f):=(\dual{\bundle})^\ast(\anchor(e)f)\quad\text{and}\quad \Lambda_{E^\ast}(\dd{}\iota_e,\dd{}\iota_{e'}):=\iota_{[e,e']}\ .
\end{equation}
Here $\iota_e$, $\iota_{e'}$ and $\iota_{[e,e']}$ are linear functions on $E^\ast$ induced by sections $e,e',[e,e']\in\Sec(E)$, respectively (i.e. $\iota_e(\xi)=\<\xi,e>$); and $f\in\Cf(M)$ is an arbitrary smooth function on $M$.
\end{proposition}
Let us remark that the above bi-vector $\Lambda_{E^\ast}$ is \emph{linear} in the sense that the related map
$$\wt{\Lambda_{E^\ast}}: \T^\ast E^\ast\lra \T E^\ast\ ,$$
defined as $\wt{\Lambda_{E^\ast}} 
(\theta):=\theta\lrcorner \Lambda_{E^\ast}$, is a DVB morphism. Lie algebroids are characterized by the condition that $\Lambda_{E^\ast}$ is a linear Poisson structure.

\section{Geometric mechanics via Tulczyjew triples}
\label{sec:TT}

\paragraph{The Hamiltonian side} 
The construction of the Hamiltonian dynamics on the dual bundle $\dual{\bundle}:E^\ast\ra M$ of an algebroid makes a direct use of the DVB morphism $\wt{\Lambda_{E^\ast}}$ introduced above. Given a \emph{Hamiltonian}, i.e. a smooth function $H\in\Cf(E^\ast)$ we define the related \emph{Hamiltonian vector field} $\X_H\in\mathfrak{X}(E^\ast)$ by
\begin{equation}
\label{eqn:ham_vf}
\X_H:=\wt{\Lambda_{E^\ast}}(\dd{} H)=\dd{} H\lrcorner\Lambda_{E^\ast} \ .
\end{equation}
Integral curves of this field are trajectories of our system, i.e. $\xi:[0,T]\ra E^\ast$ is a \emph{Hamiltonian trajectory} if 
$$\dot \xi(t)=\wt{\Lambda_{E^\ast}}(\dd{} H(\xi(t))) \quad\text{for every $t\in[0,T]$.}$$

\paragraph{The Lagrangian side}
The Lagrangian dynamic on an algebroid can also be understood in the language of natural DVB morphisms. By composing the canonical DVB isomorphism $\RR E:\T^\ast E\ra\T^\ast E^\ast$ with $\wt{\Lambda_{E^\ast}}:\T^\ast E^\ast\ra \T E^\ast$ one obtains a DVB morphism $\eps E:=\RR E\circ\wt{\Lambda_{E^\ast}}:\T^\ast E\ra \T E^\ast$, which, in fact, provides an alternative characterization of the algebroid structure on $\bundle:E\ra M$ -- see \cite{GU_1999}. This morphism, in the presence of a \emph{Lagrangian} -- a smooth function $L\in\Cf(E)$ -- allows us to give a geometric formulation of the Lagrangian dynamics on $\bundle$. Namely, a curve $\gamma:[0,T]\ra E$ {\new defines the \emph{phase dynamics} if}  
\begin{equation}
\label{eqn:EL}
\ddt{\dv L(\gamma(t))}=\eps E(\dd{} L(\gamma(t)))\quad\text{for every $t\in(0,T)$.}
\end{equation}
Perhaps it is best to visualize this equation on the following diagram, where solid arrows come from the geometry of the algebroid structure on $\bundle$, the dashed one comes from the Lagrangian, and the dotted one from the trajectory.  
\begin{equation}
\label{diag:EL}
\xymatrix{
&\T^\ast E\ar[rrr]^{\eps E}\ar[d]^{\cotangent E}\ar[dr]^{\vertical E} &&& \T E^\ast\ar[d]^{\tangent{E^\ast}}\\
[0,T]\ar@{.>}[r]^{\gamma}&E\ar@/^1pc/@{-->}[u]^{\dd{} L}& E^\ast\ar[rr]^{=}
 && E^\ast
}\end{equation}
Starting from $[0,T]$ we move right either by the upper road $\eps E(\dd{} L(\gamma)) $ or by the lower one $\vertical E(\dd {}L(\gamma))=\dv L(\gamma)$. Equation \eqref{eqn:EL} means that these two ways are consistent, i.e. the tangent lift of the lower path is precisely the upper one. 

Let us end by making a few remarks about equation \eqref{eqn:EL}:
\begin{itemize}
	\item In diagram \eqref{diag:EL} we pass from $E$ to $E^\ast$ by means of the \emph{Legendre map}, i.e. the vertical derivative of the Lagrangian $\vertical E(\dd{} L)=\dv L:E\ra E^\ast$. 
	\item {\new From the physical point of view} equation \eqref{eqn:EL} {\new should be understood as} the phase dynamics, i.e. as an equation for the curve of momenta $\dv L(\gamma(t))\in E^\ast$. Of course {\new it can be also viewed as the \emph{Euler-Lagrange equation}, i.e. equation for the trajectory $\gamma(t)\in E$, which is, however, of} \underline{implicit} nature, {\new as  
conditions imposed on the velocity $\dot{\gamma}(t)$ are indirect.} Only when the Legendre map is a (local) diffeomorphism one can make these conditions explicit. 
	\item Equation \eqref{eqn:EL} guarantees that a trajectory $\gamma:[0,T]\ra E$ is automatically an \emph{admissible curve}, i.e.
	$$\anchor(\gamma(t))=\ddt{\bundle(\gamma(t))}\quad\text{for every $t\in(0,T)$.}$$
	This follows from the fact that the DVB morphism $\eps E$ projects to a VB map $\anchor:E\ra\T M$ under $\cotangent E$ and $\T\dual{\bundle}$. 
	\item In general, the dynamics \eqref{eqn:EL} is \underline{not} determined by the image of the Lagrangian submanifold $\dd{} L(E)\subset \T^\ast E$ under $\eps E$. More precisely, unless the Legendre map is a local diffeomorphism, condition $\ddt \dv L(\gamma(t)))\in \eps E(\dd{} L(E))$ is weaker than equation \eqref{eqn:EL}. To formulate the latter we need to know the actual composition $\dd{}L \circ\eps E:E\ra \T E^\ast$, not only its image. This is not the case for the Hamiltonian dynamics, where to every element in $E^\ast$ corresponds precisely one element in $\wt{\Lambda_{E^\ast}}(\dd{}H(E^\ast))\subset\T E^\ast$. 
		\item The last observation explains why Lagrangian and Hamiltonian formalisms are non-equivalent for non-regular Lagrangians. Another argument: in a general case the Lagrangian submanifold $\dd{}L(E)\subset \T^\ast E$ is not mapped via $\RR E$ to a submanifold of the form $\dd{} H(E^\ast) \subset \T^\ast E^\ast$ -- cf. Lemma~\ref{lem:L_to_EL}. 	
	\item Equation \eqref{eqn:EL} {\new describes a true dynamics} on an algebroid, i.e. it not only looks similar to the standard Euler-Lagrange equation in local coordinates, but it is actually an equation for critical trajectories of a naturally defined action functional. In this context the appearance of the DVB morphism $\eps E$ in \eqref{eqn:EL} is not surprising -- it can be interpreted as a dual map to a relation which describes admissible variations (see \cite{GG_2008} for details).  
\end{itemize} 

\paragraph{The full Tulczyjew triple} 
Both Lagrangian and Hamiltonian formalisms discussed above can be represented on a single commutative diagram, with the Lagrangian side on the left, Hamiltonian side on the right and the phase dynamics in the middle. All the upper arrows are DVB morphisms. The (Lagr) and (Ham) solid legs of the diagram are isomorphic by means of $\RR E$. The dashed maps are induced by the Lagrangian and the Hamiltonian functions and they are responsible for generating the respective dynamics. 
\begin{equation}
\label{diag:TT}
\xymatrix{
\T^\ast E\ar[rrr]^{\eps E}\ar[d]^{\cotangent E}\ar[dr]^{\vertical E}\ar@/^2pc/[rrrrrr]_{\RR E}^{\simeq} &&& \T E^\ast\ar[d]^{\tangent{E^\ast}}
&&&\T^\ast E^\ast\ar[lll]_{\wt\Lambda_{E^\ast}}\ar[d]_{\cotangent{E^\ast}}\\
E\ar@{-->} @/^1pc/[u]^{\dd{} L}& E^\ast\ar[rr]^{=} && E^\ast&&& E^\ast\ar[lll]_{=}\ar@/_1pc/ @{-->}[u]_{\dd{}H} \\
(Lagr) &&& (Dyn) &&& (Ham)
}\end{equation}
 An analogous diagram for the special case of the tangent algebroid $E=\T M$ was discovered by Tulczyjew \cite{Tulczyjew_1976a,Tulczyjew_1976b}, hence the name -- \emph{Tulczyjew triple}.

\section{A comparison with the prolongation approach}
\label{sec:comparison}
\subsection{On prolongations}
\label{ssec:prolong}
Let us briefly recall the construction of a prolongation of an algebroid and its basic properties.

\paragraph{The $E$-tangent bundle of a fibration}
Let $(\bundle:E\ra M,\anchor,[\cdot,\cdot])$ an AL algebroid and let $\fibration:\PP\ra M$ be a fibration. By the \emph{$E$-tangent bundle of $\fibration$} we shall understand the bundle $\lift{\bundle}{\PP}:\TT \PP E\ra \PP$, where $\TT \PP E:=\{(e,X)\in E_{\fibration(p)}\times \T_p\PP\ |\ \anchor(e)=\T\fibration(X),\ p\in\PP \}$ is the pull-back space and $\lift{\bundle}{\PP}:=\liftanch{\anchor}{\PP}\circ\tangent{\PP}$ is the natural projection:
\begin{equation}
\label{eqn:TPE}
\xymatrix{\TT \PP E \ar@{-->}[rr]^{\liftanch{\anchor}{\PP}}\ar@{-->}[d]_{\down{\fibration}E} \ar@/^1.5pc/[rrrr]^{\lift{\bundle}{\PP}}&& \T \PP\ar[d]^{\T\fibration}\ar[rr]^{\tangent{\PP}}&& \PP\\
E\ar[rr]^{\anchor}&& \T M\ .
}
\end{equation}
 Here $\liftanch{\anchor}{\PP}:\TT\PP E\ra \T\PP$ and $\down{\fibration}E:\TT\PP E\ra E$ are the projections closing the diagram. 

The fiber $(\TT\PP E)_p$ at $p\in\PP$ consists of pairs $(e,X)\in E_{\fibration(p)}\times \T_p\PP$ such that $\anchor(e)=\T\fibration(X)$. By $\I \PP:=(\down{\fibration}E,\liftanch{\anchor}{\PP}):\TT\PP E\hookrightarrow E\times_M\T \PP$ we shall denote the related canonical inclusion. In fact, the $E$-tangent bundle of $\fibration$ is naturally a DVB, with the latter inclusion being a DVB morphism. 

Observe further, that $\liftanch{\anchor}{\PP}$ makes $\lift{\bundle}{\PP}$ an anchored bundle. We shall see shortly that $\lift{\bundle}{\PP}$ carries a natural AL algebroid structure inherited from the AL algebroid structure on $\bundle$, such that $\liftanch{\anchor}{\PP}$ is the corresponding anchor map, and that $\down{\fibration}{E}$ is an algebroid morphism. To show this we will first take a look at the dual bundle $\dual{(\lift{\bundle}{\PP})}:(\TT\PP E)^\ast\ra \PP$. 

\paragraph{The dual of the $E$-tangent bundle}

The construction of the dual bundle of $\lift{\bundle}{\PP}$ can be easily performed by applying the duality functor to diagram~\eqref{eqn:TPE}. That is, the total space $(\TT\PP E)^\ast$ is the following push-out\footnote{Below the arrows actually represent relations not maps -- in general a dual of a VB morphism, such as $\dual{\anchor}$ and $\dual{(\T\fibration)}$ is not a map, but a differential relation linear on fibers.}
\begin{equation}
\label{eqn:TPE_dual}
\xymatrix{(\TT \PP E)^\ast  && \T^\ast \PP \ar@{-->}_{\dual{(\liftanch{\anchor}{\PP})}}[ll]\\
E^\ast\ar@{-->}[u]_{\dual{(\down{\fibration}E)}} && \ar@{->}[ll]_{\dual{\anchor}}\T^\ast M\ar@{->}[u]_{\dual{(\T\fibration)}}\ .
}
\end{equation}
We may thus represent elements of the fiber $(\TT \PP E)^\ast_p$ as sums $\dual{(\down{\fibration}E)} \xi+\dual{(\liftanch{\anchor}{\PP})}\theta$, where $\xi\in E^\ast_{\fibration(p)}$ and $\theta\in\T^\ast_p\PP$. The pairing between such an element and a vector $(e,X)\in(\TT \PP E)_p$ reads simply as
$$\<\dual{(\down{\fibration}E)} \xi+\dual{(\liftanch{\anchor}{\PP})}\theta,(e,X)>=\<\xi,\down{\fibration}E(e,X)>+\<\theta,\liftanch{\anchor}{\PP}(e,X)>=\<\xi,e>+\<\theta,X>\ .$$
In other words, the dual of the canonical inclusion $\dual{\T{\PP}}:E^\ast_M\times\T^\ast \PP\ra (\TT \PP E)^\ast$ is an epimorphism with the kernel consisting of elements of the form $(\anchor^\ast \alpha,-(\T\fibration)^\ast(\alpha))$ for $\alpha\in\T^\ast M$.

\paragraph{The algebroid structure on $\TT\PP E$}
As was mentioned above $\lift{\bundle}{\PP}:\TT\PP E\ra\PP$ carries the canonical algebroid structure. We choose the following (non-standard) characterization of this structure which in our opinion nicely clarifies the idea standing behind the prolongation and emphasizes the role of the almost-Lie condition {\new(cf. \cite[Ex. 2.30]{MJ_MR_aha_2017} for another proof emphasizing the role of the AL condition).} 
\begin{lemma}
\label{lem:prolong}
Let $(\bundle:E\ra M,\anchor,[\cdot,\cdot])$ be an AL algebroid structure, and let $\fibration:\PP\ra M$ be a fibration. Then there is a unique AL algebroid structure on $\lift{\bundle}{\PP}:\TT\PP E\ra\PP$ such that maps $\liftanch{\anchor}{\PP}:\TT\PP E\ra\T \PP$ and $\down{\fibration}E:\TT\PP E\ra E$ are algebroid morphisms.
\end{lemma}
\begin{proof}[Sketch of the proof]The idea of the proof is quite simple -- we will show that our assumptions on $\liftanch{\anchor}{\PP}$ and $\down{\fibration}E$ fully determine the degree-one derivation $\dd{\TT \PP E}$ on the space of $\TT\PP E$-forms (cf. Proposition~\ref{prop:alg_dif}).  

For every smooth functions $f\in\Cf(\PP)$ we have $f=\dual{(\liftanch{\anchor}{\PP})} f$ (base functions on $\lift{\bundle}{\PP}:\TT\PP E\ra \PP$ and $\tangent P:\T\PP\ra \PP$ are the same). Now since we want $\liftanch{\anchor}{\PP}$ to be an algebroid morphism necessarily (cf. Definition~\ref{def:algebroid_morphism})
$$\dd{\TT \PP E}f=\dd{\TT \PP E}\left[\dual{(\liftanch{\anchor}{\PP})} f\right]=\dual{(\liftanch{\anchor}{\PP})}(\dd{} f)\ .$$
Thus, according to formula \eqref{eqn:cartan_formula}, $\liftanch{\anchor}{\PP}$ must be the anchor map of the algebroid structure on $\TT\PP E$. 

Further, since $\dual{(\I \PP)}:E^\ast\times_M\T^\ast P\ra(\TT\PP E)^\ast$ is an epimorphism, it induces a natural epimorphism at the level of sections
$$\dual{(\I\PP)}:\Cf(P)\otimes\Sec(E^\ast)\times\Omega^1(P)\ra \Sec((\TT\PP E)^\ast)\ .$$
It follows that every one-$\TT \PP E$-form $\Xi$ can be presented (non-uniquely!) as
\begin{equation}
\label{eqn:Xi_decompostion}
\Xi=\sum_i f_i\cdot\dual{(\down{\fibration}E)}e^i+\dual{(\liftanch{\anchor}{\PP})}\theta\ ,
\end{equation}
for some functions $f_i\in\Cf(P)$, sections $e^i\in \Sec(E^\ast)$ and a one-form $\theta\in\Omega^1(\PP)$.  
 
Since we require that $\down{\fibration}E$ and $\liftanch{\anchor}{\PP}$ are algebroid morphisms, by Definition~\ref{def:algebroid_morphism} we have to define 
\begin{equation}
\label{eqn:d_TT_PP_E}
\begin{split}
\dd{\TT \PP E}\Xi:=&
\sum_i\dd{\TT \PP E}\left[f_i\dual{(\down{\fibration}E)} e^i\right]+\dd{\TT \PP E}\left[\dual{(\liftanch{\anchor}{\PP})}\theta\right]=\\
&\sum_i\dd{\TT\PP E} f_i\wedge\dual{(\down{\fibration} E)}e^i+\sum_if_i\dual{(\down{\fibration}E)} \left[\dd E e^i\right]+\dual{(\liftanch{\anchor}{\PP})}\left[\dd{}\theta\right]=\\
&\sum_i\dual{(\liftanch{\anchor}{\PP})}\dd{} f_i\wedge\dual{(\down{\fibration} E)}e^i+\sum_if_i\dual{(\down{\fibration}E)} \left[\dd E e^i\right]+\dual{(\liftanch{\anchor}{\PP})}\left[\dd{}\theta\right]\ .
\end{split}
\end{equation}
As we see, the derivation $\dd{\TT\PP E}$ is fully determined by $\liftanch{\anchor}{\PP}$, $\dd E$ and the de Rham derivative $\dd{}$ on $\PP$. It suffices to check that the value of $\dd{\TT\PP E}\Xi$ does not depend on the choice of the decomposition \eqref{eqn:Xi_decompostion}. This is where axiom \eqref{eqn:AL} comes into play. Computational details can be found in Appendix~\ref{par:proof_prolong}. 
\end{proof}

When referring to bundle $\lift{\bundle}{\PP}:\TT {\PP}E\ra\PP$ with the above algebroid structure, we shall use the term \emph{the $\PP$-prolongation of an algebroid $E$}.

\paragraph{Morphism of prolongations}
Maps between fibrations induce natural algebroid morphisms between the related prolongations of a given algebroid.
\begin{lemma}\label{lem:morphism_prolong}
Let $(\bundle:E\ra M,\anchor,[\cdot,\cdot])$ be an AL algebroid, let $\fibration:\PP\ra M$ and $\fibration':\PP'\ra M$ be fibrations, and let $\psi:\PP\ra\PP'$ be a morphism of these fibrations over $\id_M$. Define a map $\morphlift{\psi}E:\TT\PP E\ra \TT{\PP'}E$ as the universal closure of the diagram
$$\xymatrix{
\TT \PP E\ar[d]_{\down{\fibration}E}\ar[rr]^{\liftanch{\anchor}{\PP}}\ar@/^1.5pc/@{-->}[rrrr]^{\morphlift{\psi}E} && \T \PP\ar@/^1.5pc/[rrrr]^{\T\psi} \ar[d]^{\T\fibration}&&\TT{\PP'} E\ar[d]_{\down{\fibration'}E}\ar[rr]^{\liftanch{\anchor}{\PP'}}  && \T \PP' \ar[d]^{\T\fibration'}\\
E\ar[rr]^{\anchor}\ar@/_1.5pc/[rrrr]^{=}&&\T M&& E\ar[rr]^{\anchor} && \T M\ ,
}$$
i.e. $\morphlift{\psi}E(e,X)=(e,\T\psi(X))$ for every $(e,X)\in\TT\PP E$. 

Then $\morphlift{\psi}E$ is an algebroid morphism.
\end{lemma}
\begin{proof}
This follows immediately from the fact that the tangent map $\T\psi:\T\PP\ra\T\PP'$ is a morphism of tangent algebroids and the characterization of algebroid prolongations from Lemma~\ref{lem:prolong}.  
\end{proof}

\subsection{Hamiltonian mechanics}
\label{ssec:hamilton}
It will be easier to start our comparison of the Tulczyjew triple and the prolongation approaches to geometric mechanics from the Hamiltonian side of the picture. 

\paragraph{Hamiltonian mechanics in the prolongation approach}

For this end, $\lift{\bundle}{E^\ast}:\TT {E^\ast}E\ra E^\ast$, the $E^\ast$-prolongation of the algebroid $\bundle:E\ra M$ will be of crucial importance. It turns out that the bundle $\lift{\bundle}{E^\ast}$ is equipped with the canonical non-degenerate 2-form $\Omega_{E}\in \Lambda^2((\TT {E^\ast} E)^\ast)$ defined by 
$$\Omega_E:=-\dd{\TT {E^\ast}E}\mu_E\ ,$$
where $\mu_E\in \Sec((\TT{E^\ast}E)^\ast)$, given by $\mu_E(\xi)(e,X):=\<\xi,e>$, is called the \emph{tautological 1-form}. Let now $H\in\Cf(E^\ast)$ be a smooth function (the \emph{Hamiltonian}). The related \emph{Hamiltonian section} $\Xi_H\in\Sec(\TT{E^\ast}E)$ is defined as the unique solution (recall that $\Omega_E$ is non-degenerate) of the equation
\begin{equation}
\label{eqn:ham_prolong}
\iota_{\Xi_H}\Omega_E=\dd{\TT{E^\ast}E}H\ .
\end{equation}
(Here we treat $H$ as a basic function on the prolongation $\lift{\bundle}{E^\ast}:\TT{E^\ast}E\ra E^\ast$.) The $\T E^\ast$-projection of $\Xi_H$, i.e. $\liftanch{\anchor}{E^\ast}(\Xi_H)$, defines the Hamiltonian dynamics. 

\paragraph{Tulczyjew triple interpretation}
Let us take a second look at equation~\eqref{eqn:ham_prolong}. On the right-hand side we have an algebroidal derivation of a function. From~\eqref{eqn:cartan_formula} we conclude that for any $(e,X)\in\TT {E^\ast}E$ we have 
$$\dd{\TT{E^\ast}E}H(e,X)=\liftanch{\anchor}{E^\ast}(e,X) H=X(H)=\<\dd{} H,X>\,$$
where in the last expression $H$ is again understood as a function on $E^\ast$, not the whole $\TT{E^\ast}E$. In other words, we have $\dd{\TT{E^\ast}E}H=\dual{(\I{E^\ast})}(0,\dd{} H)\in(\TT{E^\ast}E)^\ast$. 
Next note that $\Omega_E$ defines a VB isomorphism $\wt{\Omega_E}:\TT{E^\ast}E\ra(\TT{E^\ast} E)^\ast$. Now the Hamiltonian section can be interpreted as the image of $\dual{(\I{E^\ast})}(0,\dd{} H)$ under the inverse of this isomorphism:
\begin{equation}
\label{eqn:ham_prolong_new}
\Xi_H=\left(\wt{\Omega_E}\right)^{-1}\dual{(\I{E^\ast})}(0,\dd{} H)\ . 
\end{equation}
The following result shows that under the canonical identification ${\I{E^\ast}}$, the isomorphism $(\wt{\Omega_E})^{-1}$ is just a simple extension of the map $\wt{\Lambda_{E^\ast}}$ induced by the linear bi-vector structure $\Lambda_{E^\ast}$ on $E^\ast$ (cf. Proposition~\ref{prop:alg_bi_vect}). 
\begin{theorem}\label{thm:hamilton} The composition
$$\xymatrix{E^\ast\times_M\T^\ast E^\ast \ar[r]^>>>>>{\dual{(\I{E^\ast})}}&(\TT{E^\ast}E)^\ast\ar[rr]^{(\wt{\Omega_E})^{-1}}&&\TT{E^\ast}E \ar[r]^>>>>{\I{E^\ast}}&E\times_M\T E^\ast\ ,
}$$
reads as
\begin{equation}
\label{eqn:omega_stripped}
(\xi,\theta)\longmapsto (\vertical{E^\ast}(\theta),\wt{\Lambda_{E^\ast}}(\theta)-\V\xi)\ .
\end{equation}
Here $\vertical{E^\ast}:\T^\ast E^\ast\ra E$ is the canonical projection to the second leg of the DVB $\T^\ast E^\ast$; $\wt{\Lambda_{E^\ast}}:\T^\ast E^\ast\ra \T E^\ast$ is a DVB morphism corresponding to the linear bi-vector structure on $E^\ast$; and $\V \xi$ stands for a vertical vector field $E^\ast\ni\eta\mapsto \V_\eta\xi=\frac{\dd{}}{\dd{} t}\big|_{t=0}(\eta+t\xi)\in\T_\eta E^\ast$.  

It follows that the Hamiltonian section for $H\in\Cf(E^\ast)$ reads as
\begin{equation}
\label{eqn:hamilton_stripped}
\Xi_H=\left(\vertical{E^\ast}(\dd{}H), \wt{\Lambda_{E^\ast}}(\dd{}H)\right)=(\dv H,\X_H)\ .
\end{equation}
\end{theorem}
The proof is just a matter of computation -- see Appendix~\ref{proof:hamilton}. Actually, the coordinate formula for $\wt{\Omega_E}$ can be found in \cite{LMM_2005}.

Let us summarize our considerations. Formula~\eqref{eqn:hamilton_stripped} ensures us that the Hamiltonian evolutions obtained in both approaches are the same, i.e. $\X_H=\liftanch{\anchor}{E^\ast}(\Xi_H)$. In fact in the prolongation approach we additionally calculate the vertical derivative of the Hamiltonian. Moreover, formula \eqref{eqn:omega_stripped} in Theorem~\ref{thm:hamilton} shows how the Hamiltonian side of the Tulczyjew triple (i.e., the DVB map $\wt{\Lambda_{E^\ast}}$) ``sits'' inside the 2-form $\Omega_E$. In fact, we may think of $(\wt{\Omega_E})^{-1}$ as an extension of $\wt{\Lambda_{E^\ast}}$.

\begin{remark}
Let us remark that if we change the object which generates the dynamics in~\eqref{eqn:ham_prolong_new} by adding a ``force field''; i.e. instead of $(0,\dd{}{H})$ we take $(\xi,\dd{}H)$, where $\xi:E^\ast\to E^\ast$ is a map over $\id_M$; then the resulting dynamics changes to $\X_H-\V\xi$. This can be interpreted as a Hamiltonian system with external forces.  
\end{remark} 

\subsection{Lagrangian mechanics}
\label{ssec:lagrange}

The construction of the Lagrangian formalism in the prolongation approach tries to mimic the structures present in the Hamiltonian formalism, yet this time one shall work with  $\lift{\bundle}E:\TT E E\ra E$ -- the $E$-prolongation of $\bundle:E\ra M$ -- instead of $\lift{\sigma}{E^\ast}$. Given a \emph{Lagrangian} $L\in\Cf(E)$ one aims to construct a 2-form $\omega_L\in\Sec(\Lambda^2(\TT E E)^\ast)$ (it will depend on $L$) and then formulates the dynamics in a form similar to \eqref{eqn:ham_prolong}. The original construction of $\omega_L$ from \cite{Martinez_2001} is quite complicated as it generalizes an analogous construction from Klein's paper \cite{Klein_1962} which makes extensive use of the canonical almost-tangent structure on the tangent bundle. Instead, we prefer the following simplified approach.

\paragraph{Lagrangian mechanics in the prolongation approach} 
We shall use the \emph{Legendre map} related with $L$, i.e. $\lambda_L:=\dv L:E\ra E^\ast$. It gives rise to the related map on prolongations $\morphlift{\lambda_L}E:\TT E E\ra \TT {E^\ast}E$. Now we can define the \emph{energy} $\EL\in\Cf(E)$ and the desired 2-form $\omega_L\in\Sec(\Lambda^2(\TT E E)^\ast)$ by 
$$\EL(e):=\<\lambda_L(e),e>-L(e)\quad \text{and}\quad \omega_L:=(\morphlift{\lambda_L}E)^\ast \Omega_E\ .$$
We say that a curve $\gamma:[0,T]\ra E$ is a \emph{solution of the Euler-Lagrange equations} if 
\begin{equation}
\label{eqn:EL_prolong}
\iota_{(\gamma(t),\dot\gamma(t))}\omega_L=\dd{\TT EE}\EL(\gamma(t))\quad\text{for every $t\in(0,T)$.}
\end{equation}
So the basic idea here is to pull-back the symplectic form $\Omega_E$ from $\TT {E^\ast} E$ to $\TT EE$ by means of the Legendre map and to change the generating object of the dynamics from the Lagrangian $L$ to the energy $\EL$. This is completely analogous to the original Weinstein's treatment of regular Lagrangian systems in \cite{Weinstein_1996}. Quite surprisingly here this approach works well also in a non-regular case. However, \emph{a priori} it is not clear if equation \eqref{eqn:EL_prolong} describes the same curves as equation \eqref{eqn:EL}. 
\paragraph{Relation with the Tulczyjew triple} 
Now we shall study the geometry behind equation \eqref{eqn:EL_prolong} and its relations with the Tulczyjew triple approach. 

By the same argument as in the Hamiltonian case, we actually have $\dd{\TT EE}\EL=(\I E)^\ast(0,\dd{}\EL)$. On the other hand, by construction  the morphism $\wt{\omega_L}:\TT EE\ra (\TT E E)^\ast$, naturally induced by the two-form $\omega_L$, is defined as the dashed arrow closing the following diagram of maps between prolongations (for convenience we also marked two legs of the respective pull-back or push-out).
\begin{equation}
\label{diag:omega_L}
\xymatrix{&&\TT{E^\ast}E\ar[dl]\ar[dr] &&& \ar[lll]_{(\wt{\Omega_E})^{-1}}^{\simeq}(\TT {E^\ast}E)^\ast\ar[dd]^{(\morphlift{\lambda_L} E)^\ast}&\\
&\T E^\ast &&E&E^\ast\ar[ru]\ar[dd]^<<<<<<{=}&& \T^\ast E^\ast\ar[ul]\ar[dd]^{(\T\lambda_L)^\ast}
&\\
&&\TT EE\ar[dl]\ar[uu]^{\morphlift{\lambda_L}E} \ar[dr]\ar@{-->}[rrr]^{\wt{\omega_L}}&&& (\TT E E)^\ast&\\
[0,T]\ar@{..>}[r]^{\T\gamma} \ar@/_1.5pc/@{..>}[rrr]^{\gamma}&\T E \ar[uu]^{\T\lambda_L}&&E\ar[uu]^<<<<<<{=}&E^\ast\ar[ru]&& \T^\ast E\ar[ul]_{(\liftanch{\anchor}{E})^\ast}&[0,T] \ar@{..>}[l]_{\dd{}\EL(\gamma)}\ar@{..>}@/^1.5pc/[lll]_{0}\ .
}\end{equation}
When formulating equation \eqref{eqn:EL_prolong} we basically ``feed'' this diagram by a pair of curves $\gamma\circ(0,\dd{} \EL):[0,T]\ra E\times\T^\ast E$ in the bottom-right corner; and a pair of curves $(\gamma,\dot \gamma):[0,T]\ra E\times\T E$ in the bottom-left corner (dotted arrows) and demand that the resulting curves are compatible. Now, the result below follows directly from formula \eqref{eqn:omega_stripped} and  simple diagram chasing.
\begin{proposition}
A curve $\gamma:[0,T]\ra E$ is a solution of equation \eqref{eqn:EL_prolong} if and only if for every $t\in[0,T]$ there exists an element $\theta_t\in \T^\ast E^\ast$ such that $\dd{}\EL(\gamma(t))=(\T\lambda_L)^\ast\theta_L$ and $(\gamma(t),\T\lambda_L(\dot\gamma(t)))=\left((\I{E^\ast})^\ast\circ\wt{\Omega_E}^{-1}\circ\I{E^\ast}\right)(0,\theta_t)$. By the results of Theorem~\ref{thm:hamilton} this translates as
\begin{align}
\label{eqn:EL_1}
\dd{}\EL(\gamma(t))&=(\T\lambda_L)^\ast\theta_t\ ,\\
\label{eqn:EL_2}
\gamma(t)&=\vertical{E^\ast}(\theta_t) \ ,\\
\label{eqn:EL_3}
\T\lambda_L(\dot\gamma(t))&=\wt{\Lambda_{E^\ast}}(\theta_t)\ .
\end{align} 
\end{proposition}
To see that equations \eqref{eqn:EL_1}--\eqref{eqn:EL_3} are equivalent to equation \eqref{eqn:EL} {\new (understood as an equation for $\gamma:[0,T]\ra E$)} we will need the following two lemmata. 
\begin{lemma}\label{lem:L_to_EL}
 The following diagram commutes
\begin{equation}
\label{diag:L_to_EL}
\xymatrix{\T^\ast E^\ast\ar[d]_{(\T\lambda_L)^\ast} && \T^\ast E\ar[ll]_{\RR E}\\
\T^\ast E && E\ .\ar[u]_{\dd{} L}\ar[ll]_{\dd{}\EL}
}\end{equation}
\end{lemma}
A simple proof is moved to Appendix~\ref{proof:L_to_EL}.

\begin{lemma}\label{lem:theta_unique}
An element $\theta_t:=\RR E(\dd{} L(\gamma(t)))$ is the unique element in $\T^\ast E^\ast$ which satisfies \eqref{eqn:EL_1} and \eqref{eqn:EL_2}. 
\end{lemma}
\begin{proof} The fact that such a $\theta_t$ satisfies \eqref{eqn:EL_1} and \eqref{eqn:EL_2} is obvious: \eqref{eqn:EL_1} follows from the commutativity of \eqref{diag:L_to_EL}, while \eqref{eqn:EL_2} from \eqref{eqn:R_E_legs}: $\vertical {E^\ast}(\RR E(\dd{} L(\gamma(t))))=\cotangent{E}(\dd{}L(\gamma(t)))=\gamma(t)$.

Now any other element with these properties would be of the form $\theta_t+\theta'$, where $\theta'\in\T^\ast_{\lambda_L(\gamma(t))}E^\ast$ is such that  $\vertical{E^\ast}(\theta')=0$ and $\theta'\in \Ann(\T\lambda_L\big|_{\gamma(t)})$. It takes an easy coordinate calculation to show that such a $\theta'$ must be a null covector.  \end{proof}
Finally we are ready to close our considerations about the Lagrangian formalism in the prolongation approach.  
\begin{theorem}
\label{thm:lagrangian}
Equations \eqref{eqn:EL} and \eqref{eqn:EL_prolong} are equivalent. More precisely, the prolongation approach to Lagrangian mechanics ``hides'' inside itself the following diagram: 
\begin{equation}
\label{diag:lagrangian}
\xymatrix{&\T E^\ast && \T^\ast E^\ast\ar[ll]_{\wt{\Lambda_{E^\ast}}}\ar@{-->}[d]_{(\T\lambda_L)^\ast} && *+[F--]{\T^\ast E} \ar[ll]_{\RR E}\ar@/_1.5pc/[llll]_{\eps E}\\
[0,T]\ar@{..>}[r]^{\T\gamma}&\T E \ar@{-->}[u]^{\T\lambda_L}&& \T^\ast E && E\ar@{}[llu]|{\eqref{diag:L_to_EL}} \ar@{-->}[ll]_{\dd{} \EL}\ar@{-->}[u]^{\dd{}L}
&[0,T]\ .\ar@{..>}[l]_{\gamma}}
\end{equation}
\end{theorem}
\begin{proof} By Lemma~\ref{lem:theta_unique} we know that under \eqref{eqn:EL_1} and \eqref{eqn:EL_2}, necessarily $\theta_t=\RR E(\dd{} L(\gamma(t)))$ and thus \eqref{eqn:EL_3} reads as
$$\ddt \dv L(\gamma(t))=\T\lambda_L(\dot\gamma(t))=\wt{\Lambda_{E^\ast}}(\RR E(\dd{} L(\gamma(t))))=\eps E(\dd{}L(\gamma(t)))\ ,$$
and so \eqref{eqn:EL} holds. 

Diagram \eqref{diag:lagrangian} is obtained from \eqref{eqn:omega_stripped} after composing diagram \eqref{diag:omega_L} with diagram \eqref{diag:L_to_EL}. By definition $\eps E:=\wt{\Lambda_{E^\ast}}\circ\RR E$, which explains the up-most arrow.  
\end{proof}
As before, in diagram \eqref{diag:lagrangian} the solid arrows are given by the geometry, the dashed ones by the Lagrangian and the dotted ones by the actual curve. As we see, the Lagrangian side of the Tulczyjew triple \eqref{diag:TT} ``sits'' inside the prolongation approach. Its presence is \emph{a priori} not that obvious, as it requires an addition of an extra $\T^\ast E$-nod (the framed entry in \eqref{diag:lagrangian}). In fact in the Tulczyjew triple context it is much more natural to view \eqref{diag:lagrangian} as the Lagrangian side of the Tulczyjew triple (upper row and the $\dd{} L$ arrow) sandwiched between two non-canonical maps: $\T\lambda_L$ and $(\T\lambda_L)^\ast$. This explains the presence of the $\dd{} \EL$-arrow as the natural closure of the diagram (this follows from Lemma~\ref{lem:L_to_EL}).

\subsection{The prolongation Tulczyjew triple}
\label{ssec:prolong_tt}
The Hamiltonian formalism in the prolongation setting is quite elegant. On the other hand, the construction of the Lagrangian formalism featured the usage of the Legendre map, which lead to a somehow unnatural (non-canonical) construction \eqref{diag:lagrangian}. Such a dichotomy between the Lagrangian and Hamiltonian formalisms is, however, not present in the  Tulczyjew triple approach, where the passage from the Hamiltonian to the Lagrangian formalism can be described as follows: we transform the DVB $\T^\ast E^\ast$ to $\T^\ast E$ via the canonical isomorphism $\RR E$, and then change the generating map of the dynamics from $\dd{} H:E^\ast\ra \T^\ast E^\ast$ to $\dd{} L:E\ra \T^\ast E$. There are \emph{a priori} no obstructions to repeat this procedure in the prolongation approach. This will eventually lead to an alternative formulation of the Lagrangian formalism in the language of prolongations. 

\paragraph{The canonical isomorphism on prolongations} 
We shall begin by adopting the canonical isomorphism $\RR E$ to the prolongation context. 

\begin{lemma_def}\label{lem:prolong_isom}
Spaces $\dual{(\TT E E)}$ and $\dual{(\TT{E^\ast}E)}$ are canonically isomorphic via the trivial extension of $\RR E$.
\end{lemma_def}
\begin{proof} Recall from \eqref{eqn:TPE_dual} that the spaces in consideration are characterized as push-outs of $E^\ast$ and, respectively, $\T^\ast E$ and $\T^\ast E^\ast$. Now, the latter spaces are isomorphic via $\RR E$, which additionally preserves the inclusions $\T^\ast M\hookrightarrow \T^\ast E$ and $\T^\ast M\hookrightarrow \T^\ast E^\ast$ (these are the \emph{cores} of the respective DVB structures -- see \cite{Konieczna_Urbanski_1999}). As a result we obtain an isomorphism of push-outs $\wt{\RR E}:\dual{(\TT E E)}\ra \dual{(\TT {E^\ast}E)}$ canonically closing the following diagram:
$$\xymatrix{
\dual{(\TT E E)}\ar@/^1.5pc/@{-->}[rrrr]^{\wt{\RR E}} && \T^\ast E \ar@{-|>}_{\dual{(\liftanch{\anchor}{{E}})}}[ll]\ar@/^1.5pc/[rrrr]^{\RR E} &&\dual{(\TT{E^\ast} E)}  && \T^\ast E^\ast \ar@{-|>}_{\dual{(\liftanch{\anchor}{E^\ast})}}[ll]\\
E^\ast\ar@/_1.5pc/[rrrr]^{=}\ar@{-|>}[u]_{\dual{(\down{\bundle}E)}}&&\T^\ast M \ar@/_1.5pc/[rrrr]^{=} \ar@{-|>}[ll]_{\dual{\anchor}}\ar@{-|>}[u]_{\dual{(\T{\bundle})}}&& E^\ast \ar@{-|>}[u]_{\dual{(\down{(\dual{\bundle})}E)}}&& \T^\ast M\ar@{-|>}[ll]_{\dual{\anchor}}\ar@{-|>}[u]_{\dual{(\T\dual{\bundle})}}\ .
}$$
\end{proof}

\paragraph{Another formulation of mechanics}
Our previous results immediately lead to the following.
\begin{proposition}
\label{prop:lagrangian} 
Define $\wt{\eps E}:\dual{(\TT E E)}\ra \TT{E^\ast}E$ as the composition $\wt{\eps E}:=\wt {\RR E}\circ(\wt{\Omega_E})^{-1}$. Under canonical identifications $\I{E}$ and $\I{E^\ast}$ the map
$$\xymatrix{
E^\ast\times_M\T^\ast E \ar[r]^>>>>>{\dual{(\I{E})}}&(\TT EE)^\ast\ar[rr]^{\wt{\eps E}}&&\TT{E^\ast}E \ar[r]^>>>>{\I{E^\ast}}&E\times_M\T E^\ast\ ,
}$$
reads as
\begin{equation}
\label{eqn:epsilon_stripped}
(\xi,\theta)\longmapsto (\cotangent{E}(\theta),\eps E(\theta)-\V\xi)\ .
\end{equation}
\end{proposition}
Obviously, we can use mappings $(\wt{\Omega_E})^{-1}$, $\wt{\eps E}$ and $\wt{\RR E}$ to build a \emph{prolongation analog of the Tulczyjew triple} and use it in the geometric formulation of mechanics analogous to the one in \eqref{diag:TT}. This construction appeared already in \cite{LMM_2005}, however only as a generalization of the classical Tulczyjew triple, and not a framework of geometric mechanics. 
\begin{corollary}
\label{cor:prolong_tt}
The \emph{prolongation Tulczyjew triple} is the following commutative diagram consisting in the upper row of DVB isomorphisms. 
$$
\xymatrix{
&\dual{(\TT E E)} \ar@/^2pc/[rrrrrr]^{\wt{\RR E}}_{\simeq}\ar[rrr]^{\wt{\eps E}}_{\simeq}&&&\TT{E^\ast}E\ar[dl]_{\liftanch{\anchor}{E^\ast}}\ar[d]&&&\dual{(\TT{E^\ast}E)}\ar[lll]_{(\wt{\Omega_E})^{-1}}^{\simeq}\\
*+[F]{\T^\ast E}\ar[d]\ar[ur]\ar[dr]^{\vertical E} & E^\ast \ar[u] &&*+[F]{\T E^\ast}\ar[d] & E && E^\ast \ar[ur]&*+[F]{\T^\ast E^\ast}\ar[u]\ar[d]\\
E \ar@{-->}@/^1pc/[u]^{\dd{} L}\ar@{-->}[ur]^<<<<<<{0}& E^\ast\ar[rr]^{=} &&{E^\ast}&&&& E^\ast\ar[llll]_{=}\ar@{-->}@/_1pc/[u]_{\dd{}H}\ar@{-->}[ul]_{0}\\
&(Lagr) &&& \quad(Dyn) &&& (Ham)
}$$
Formulas~\eqref{eqn:omega_stripped} and \eqref{eqn:epsilon_stripped} show that this structure is nothing more but a simple extension of the Tulczyjew triple of the initial algebroid \eqref{diag:TT} (framed entries). A (disputable) gain is the fact that the new triple consists of DVB isomorphisms, not just DVB morphisms. 

The Lagrangian dynamics is generated by a map $\wt{\dd{} L}:=(\dd{} L,0):E\ra \T^\ast E\times_M E^\ast$. That is a curve $\gamma:[0,t]\ra E$ is a solution of the EL equation if the composition $\liftanch{\anchor}{E^\ast}(\wt{\eps E}(\wt{\dd{} L}(\gamma)))$ is the tangent prolongation of $\dv L(\gamma)$. In an analogous manner, the Hamiltonian dynamics is generated by $\wt{\dd{} H}:=(\dd{}H,0):E^\ast\ra \T^\ast E^\ast\times_ME^\ast$. 
\end{corollary}

\section{Conclusions}\label{sec:concl}
Our aim in this work was to discuss a geometric relation between the prolongation and Tulczyjew triple approaches to Geometric Mechanics. Our study (in Theorem~\ref{thm:hamilton}) reveals that in the Hamiltonian case the prolongation approach is rooted on a simple extension of the DVB morphism $\wt{\Lambda_{E^\ast}}:\T^\ast E^\ast\ra \T E^\ast$, which constitutes the Hamiltonian side of the Tulczyjew triple \eqref{diag:TT}. In the Lagrangian case (Theorem~\ref{thm:lagrangian}) we have a similar situation with the DVB morphism $\eps E:\T^\ast E\ra \T E^\ast$, but the latter is additionally sandwiched between the lifts of the Legendre map. (The latter step can be avoided -- see Proposition~\ref{prop:lagrangian}.) Since the Tulczyjew triple is present inside the prolongation approach, and not \emph{vice versa}, it is clear that the former approach is {\new primary to} the latter. Other arguments also favor the Tulczyjew triple approach:
\begin{itemize}
	\item The main gain of using the prolongation approach is that it allows to formulate the equations in (pre)symplectic terms. Yet this comes at a price of extending the space, say, from $\T E^\ast$ to $\TT {E^\ast} E$ , and at the end of the day we need to project the (pre)symplectic dynamics back to $\T E^\ast$, so actually the (pre)symplecicity is lost.				
	\item To build an algebroid structure on the prolongation (which is used in equations \eqref{eqn:ham_prolong} and \eqref{eqn:EL_prolong}) we need the initial algebroid structure to be of class almost-Lie. This follows directly from the proof of Lemma~\ref{lem:prolong}. On the other hand, the Tulczyjew triple approach allows for a direct generalization for a much wider class of \emph{general algebroids} -- see \cite{GGU_2006}. 
	\item	Equation~\eqref{eqn:EL} actually deconstructs the process of calculating the Euler-Lagrange equations from the variational principle. Indeed, map $\eps E$ can be interpreted as a geometric counterpart of the  procedure of extracting a generator (virtual displacement) from a variation, and the condition of the compatibility of curves $\eps E(\dd{} L(\gamma))$ and $\dv L(\gamma)$ corresponds to the procedure of integration by parts -- see \cite{GG_2008} and also \cite{MJ_MR_2014}. Such an interpretation is not present or at least not known in the case of the prolongation approach. 
	\item This may seem only as a matter of esthetics, but in the Lagrangian case the Tulczyjew triple approach seems to be simpler. Indeed, in order to obtain the equations of motion we use a canonical object $\eps E$ and a non-canonical one $\dd{} L$, whereas in the prolongation approach we need two non-canonical objects $\omega_L$ and $\dd{}E_L$.  
\end{itemize}

\section*{Acknowledgments}
\addcontentsline{toc}{section}{Acknowledgments}
This research was supported by the Polish National Science Center (grant DEC-2012/06/A/ST1/00256). {\new I would like to thank professor Janusz Grabowski for pointing out a mistake in the previous version of this manuscript.} 

\appendix
\section{Proofs of some results}

\paragraph{Details of the proof of Lemma~\ref{lem:prolong}.}\label{par:proof_prolong}
Let us prove that formula \eqref{eqn:d_TT_PP_E} provides a reasonable definition of $\dd{\TT\PP E}\Xi$.

First, note that an element $\zeta\in\Cf(P)\otimes \Sec(E^\ast)$ may have many presentations in the form $\zeta=f_i\cdot e^i$, with $f_i\in\Cf(P)$ and $e^i\in\Sec(E^\ast)$.\footnote{From now on the Einstein's summation convention is assumed.} Let us calculate an evaluation of a 2-$\TT\PP E$-form $\dd{\TT \PP E}[\dual{(\down{\fibration}E)}\zeta]=\dd{\TT\PP E}f_i\wedge\dual{(\down{\fibration} E)}e^i+f_i\dual{(\down{\fibration}E)} \left[\dd E e^i\right]$ on a pair of vectors $(e,X),(e',X')\in(\TT\PP E)_p$. With some abuse of notation let us extend these elements to local sections $(e, X),(e',X')\in\Sec(\TT\PP E)$  such that $e=\dual{\fibration}\wt e$ and $e'=\dual{\fibration}\wt e'$ for some $\wt e,\wt e'\in\Sec(E)$ (i.e. sections $e$ and $e'$ are constant on fibers of $\fibration$, so they project to sections $\wt e$ and $\wt e'$ of $\bundle$). Now
\begin{align*}
&\left[\dd{\TT\PP E}f_i\wedge\dual{(\down{\fibration} E)}e^i+f_i\dual{(\down{\fibration}E)} \left[\dd E e^i\right]\right]((e,X),(e',X'))=\\ 
&X(f_i)\<e^i,\wt e'>-X'(f_i)\<e^i,\wt e>+f_i\dd Ee^i(\wt e,\wt e')\overset{\eqref{eqn:cartan_formula}}=\\
&X(f_i)\<e^i,\wt e'>-X'(f_i)\<e^i,\wt e>+f_i\left(\anchor(\wt e)\<e^i,\wt e'>-\anchor(\wt e')\<e^i,\wt e>-\<e^i,[\wt e,\wt e']>\right)=\\
&X\left(f_i\<e^i,\wt e'>\right)-X'\left(f_i\<e^i,\wt e>\right)-\<f_ie^i,[\wt e,\wt e']>=X\<\zeta,\wt e'>-X'\<\zeta,\wt e>-\<\zeta,[\wt e,\wt e']>\ ,
\end{align*}
where the passage to the last line is due to the fact that $\T\fibration (X)=\anchor(\wt e)$ and $\T\fibration (X')=\anchor(\wt e')$. It is now clear that this expression does not depend on the chosen decomposition of $\zeta$. Moreover, we have proved the following formula:

\begin{equation}
\label{eqn:dd_prolong}
\dd{\TT \PP E}[\dual{(\down{\fibration}E)}\zeta]((e,X),(e',X'))=X\<\zeta,\wt e'>-X'\<\zeta,\wt e>-\<\zeta,[\wt e,\wt e']>\ ,
\end{equation}
where the extensions $\wt e$ and $\wt e'$ are as described above. 
\medskip

Finally, since the kernel of $\dual{(\I{\PP})}$ consists of elements of $\T^\ast M$, the decomposition 
$$\Xi=\dual{(\down{\fibration}E)}\zeta+\dual{(\liftanch{\anchor}{\PP})}\theta $$
is determined up to a $\Cf(\PP)$-span of one-forms on $M$. Thus we are left with checking if $\dual{(\down{\fibration}E)} \left[\dd E\xi\right]=-\dual{(\liftanch{\anchor}{\PP})}\left[\dd{}\theta\right]$ for $\xi=\dual{\anchor}\alpha$ and $-\theta=\dual{(\T\PP)}\alpha$ for $\alpha\in\Omega^1(M)$. For sections $(e,X),(e',X')\in \Sec(\TT\PP E)$ as before, we have
$$
\dual{(\down{\fibration}E)} \left[\dd E\dual{\anchor} \alpha\right]((e,X),(e',X))=\dd E(\dual{\anchor} \alpha)(\wt e,\wt e')\overset{\eqref{eqn:cartan_formula}}=\anchor(\wt e)\alpha(\anchor(\wt e'))-\anchor(\wt e')\alpha(\anchor(\wt e))-\alpha(\anchor[\wt e,\wt e'])\ .
$$
On the other hand, since $\T\fibration (X)=\anchor(\wt e)$ and $\T\fibration (X')=\anchor(\wt e')$, 
\begin{align*}
&\dual{(\liftanch{\anchor}{\PP})}\left[\dd{}\dual{(\T\fibration)}\alpha\right]((e,X),(e',X'))=\dd{}(\T\fibration)^\ast\alpha(X,X')=\\
&\dd{} \alpha(\T\fibration(X),\T\fibration(X'))=\dd{}\alpha(\anchor(\wt e),\anchor(\wt e'))=\anchor(\wt e)\alpha(\anchor(\wt e'))-\anchor(\wt e')\alpha(\anchor(\wt e))-\alpha([\anchor(\wt e),\anchor(\wt e')]_{\T M})\ .
\end{align*}
By the almost-Lie condition \eqref{eqn:AL} the two above expressions are equal, which shows that the derivation $\dd{\TT\PP E}$ is well-defined. Finally note that by construction the new anchor $\liftanch{\anchor}{\PP}$ is an algebroid morphism, i.e. the construed algebroid structure on $\TT\PP E$ is almost-Lie -- cf. the remark following Definition~\ref{def:algebroid_morphism}. 
\qed

\paragraph{Proof of Theorem~\ref{thm:hamilton}}\label{proof:hamilton}

First let us observe that $\mu_E=\dual{(\down{\fibration}E)}\id_{E^\ast}$, where $\id_{E^\ast}:E^\ast\ra E^\ast$ is the identity map. Since, $\Omega_E=-\dd{\TT{E^\ast}E}\mu_E$, formula \eqref{eqn:dd_prolong} gives us
\begin{equation}
\label{eqn:omega_evaluation}
\Omega_E((e,X),(e',X'))=\iota_{[\wt e,\wt e']}-X(\iota_{\wt e'})+X'(\iota_{\wt e})=\iota_{[\wt e,\wt e']}-\<\dd{}\iota_{\wt e'},X>+\<\dd{} \iota_{\wt e},X'>\ ,
\end{equation}
where $\wt e$ and $\wt e'$ are extensions of $(e,X)$ and $(e',X')$ defined as in the proof of Lemma~\ref{lem:prolong} above, and $\iota_{\wt e}(\xi)=\<\xi,\wt e>$ is the evaluation map. 

The assertion \eqref{eqn:omega_stripped} holds if and only if for every $(\xi,\theta)\in E^\ast\times_M\T^\ast E^\ast$ and every $(e',X')\in\TT{E^\ast}E$ we have
\begin{equation}
\label{eqn:omega_equivalent}
\Omega_E((\vertical{E^\ast}(\theta),\wt{\Lambda_{E^\ast}}(\theta)-\V\xi),(e',X'))=\<\xi,e'>+\<\theta,X'>\ .
\end{equation}
Actually, it is enough to check it only for covectors $\theta=\dd{}\iota_{\wt e}$ and $\theta=(\dual{\bundle})^\ast\dd{} f$ where $\wt e\in\Sec(E)$ and $f\in\Cf(M)$ are arbitrary, as (at a point) such elements span every fiber of $\cotangent{E^\ast}:\T^\ast E^\ast\ra E^\ast$. From \eqref{eqn:omega_evaluation} we get
\begin{align*}
\Omega_E((\wt e,\wt{\Lambda_{E^\ast}}(\dd{}\iota_{\wt e})-\V\xi),(e',X'))=&\iota_{[\wt e,\wt e']}-\<\dd{}\iota_{\wt e'},\wt{\Lambda_{E^\ast}}(\dd{}\iota_{\wt e})-\V\xi>+\<\dd{}\iota_{\wt e},X'>\overset{\eqref{eqn:lambda_formula}}=\\
&\iota_{[\wt e,\wt e']}-\iota_{[\wt e,\wt e']}+\<\xi,e'>+\<\dd{}\iota_{\wt e},X'>=\<\xi,e'>+\<\dd{}\iota_{\wt e},X'>\ ,
\end{align*}
and thus \eqref{eqn:omega_equivalent} holds for $\theta=\dd{}\iota_{\wt e}$. Similarly, since $\T(\dual{\bundle})(X')=\anchor(\wt e')$, we get 
\begin{align*}
\Omega_E((0,\wt{\Lambda_{E^\ast}}((\dual{\bundle})^\ast\dd{}f)-\V\xi),(e',X'))=&\iota_{[0,\wt e']}-\<\dd{}\iota_{\wt e'},\wt{\Lambda_{E^\ast}}((\dual{\bundle})^\ast\dd{}f)-\V\xi>+\<\dd{}\iota_{0},X'>\overset{\eqref{eqn:lambda_formula}}=\\
&\anchor(\wt e')f+\<\xi,e'>=\<\dd{}f,\T(\dual{\bundle})(X')>+\<\xi,e'>=\\
&\<(\dual{\bundle})^\ast\dd{}f,X'>+\<\xi,e'>\ ,
\end{align*}
proving \eqref{eqn:omega_equivalent} for $\theta=(\dual{\bundle})^\ast \dd{}f$. This ends the reasoning. \qed 

\paragraph{A proof of Lemma \ref{lem:L_to_EL}}\label{proof:L_to_EL}

We need to check if for every $a\in E$ the following equality holds $\dd{}\EL(a)=(\T\lambda_L)^\ast\RR E(\dd{}L(a))$. Equivalently, we want to prove that for every $X\in\T_aE$ we have
\begin{equation}
\label{eqn:to_prove}
\<\dd{}\EL(a),X>=\<\RR E(\dd{}L(a)),\T\lambda_L(X)>\ .
\end{equation}
Note that by \eqref{eqn:R_E_legs} we have $\RR E(\dd{} L(a))\in\T_\xi E^\ast$ where $\xi=\vertical E(\dd{} L(a))=\lambda_L(a)$. 
Now, by definition {\new (cf. Subsection~\ref{ssec:preliminaries}) the pair $(\dd{}L(a),\RR E(\dd{} L(a)))\in\T^\ast E\times\T^\ast E^\ast$ acts on $\T(E\times_ME^\ast)\simeq \T E\times_{\T M}\T E^\ast$ as the differential of the canonical pairing $\<\cdot,\cdot>:E\times E^\ast \ra \R$. That is, for any $X\in \T_a E$ and $Y\in \T_\xi E^\ast$ which have the same $\T M$-projection 
$$\<\dd{}L(a),X>+\<\RR E(\dd{}L(a)),Y>=\dd{}\<\cdot,\cdot>(X,Y)\ .$$
Taking $Y=\T\lambda_L(X)$ (check that the $\T M$-projections agree) leads to 
$$\<\RR E(\dd{}L(a)),\T\lambda_L(X)>=\dd{}\<\cdot,\cdot>(X,\T\lambda_L(X))-\<\dd{}L(a),X>\ .$$
Now, as $\EL(a)=\<a,\lambda_L(a)>-L(a)$, the right-hand-side of the above equality clearly equals $\<\dd{}\EL(a),X>$, and hence \eqref{eqn:to_prove} holds.\qed
}

\bibliographystyle{amsalpha}
\addcontentsline{toc}{section}{References}
\bibliography{bibl}
\end{document}